\documentclass[11pt]{article}

\usepackage[letterpaper,margin=2.2cm,includehead, includefoot, 
marginparwidth=1.8cm,marginparsep=.25cm]{geometry}

\usepackage{hyperref}%
\hypersetup{colorlinks=true, linkcolor=blue, breaklinks=true, urlcolor=blue}
\usepackage{doi}

\usepackage[all]{xy}
\usepackage{amsmath}
\usepackage{amscd}
\usepackage{amsthm}
\usepackage{mathrsfs}
\usepackage{amssymb}
\usepackage{tikz}
\usepackage[utf8]{inputenc}
\usepackage[yyyymmdd]{datetime}
\usepackage{graphicx}
\usepackage{subcaption}
\usepackage{lipsum}
\usepackage{multicol} \usepackage{mathtools}
\usepackage{multirow}
\usepackage{threeparttable}
\usepackage{pdfpages}
\usepackage{multicol} 
\usepackage{mathtools}
\usepackage{version,xspace}
\usepackage{url,doi}
\usepackage{wrapfig}
 
\DeclareSymbolFont{bbold}{U}{bbold}{m}{n}
\DeclareSymbolFontAlphabet{\mathbbold}{bbold}
\newcommand{\vect}[1]{\mathbbold{#1}}

\newcommand\oprocendsymbol{\hbox{$\triangle$}}
\newcommand\oprocend{\relax\ifmmode\else\unskip\hfill\fi\oprocendsymbol}

\usepackage{enumitem}

\newtheorem{theorem}{Theorem}[section]

\newtheorem{lemma}[theorem]{Lemma}
\newtheorem*{lemma*}{Lemma}
{\theoremstyle{definition}
	\newtheorem{definition}[theorem]{Definition}
	
	\newtheorem{example}[theorem]{Example}
	
	\newtheorem*{example*}{Example}
}
\numberwithin{equation}{section}
\newcommand{\bullobib}{\bibliographystyle{plainurl}\bibliography{alias,Main,FB,New}}

\title{A Tutorial on Multivariate $k$-Statistics and their Computation}
\author{Kevin D. Smith$^1$}

\date{%
	$^1$Center for Control, Dynamical Systems, and Computation\\
	University of California, Santa Barbara\\[2ex]%
	\today
}

\DeclareMathOperator{\E}{\mathbb E}
\DeclareRobustCommand{\stirling}{\genfrac\{\}{0pt}{}}

\begin{document}
	
	\maketitle
	
	\begin{abstract}
		This document aims to provide an accessible tutorial on the unbiased estimation 
		of multivariate cumulants, using $k$-statistics. We offer an explicit and 
		general formula for multivariate $k$-statistics of arbitrary order. We also prove 
		that the $k$-statistics are unbiased, using M\"obius inversion and rudimentary 
		combinatorics. Many detailed examples are considered throughout the paper. We 
		conclude with a discussion of $k$-statistics computation, including the challenge 
		of time complexity, and we examine a couple of possible avenues to improve the 
		efficiency of this computation. The purpose of this document is threefold: to 
		provide a clear introduction to $k$-statistics without relying on specialized 
		tools like the umbral calculus; to construct an explicit formula for 
		$k$-statistics that might facilitate future approximations and faster algorithms; 
		and to serve as a companion paper to our Python library \textit{PyMoments} 
		\cite{KDS:20}, which implements this formula.
	\end{abstract}

	\footnotetext[2]{This work was supported in part by the U.S. Defense Threat Reduction 
	Agency, under grant HDTRA1-19-1-0017.}
	
	\section{Introduction}
	
	Cumulants are a class of statistical moments that succinctly describe univariate and 
	multivariate distributions. Low-order cumulants are quite familiar: first-order 
	cumulants are means, second-order cumulants are covariances, and third-order 
	cumulants are third central moments. But fourth-order cumulants and larger are 
	difficult to express in terms of central or raw moments. Still, higher-order 
	cumulants have found a variety of applications, largely because they preserve the 
	intuition of central moments while also featuring desirable multilinearity and 
	additivity properties. Various applications have exploited these properties to solve 
	problems in statistics, signal processing, control theory, and other fields.
	
	This note concerns the unbiased estimation of cumulants from data. The 
	canonical unbiased estimators of cumulants, known as \textit{Fisher's 
	$k$-statistics}, or more simply \textit{$k$-statistics}, have been around since 
	Fisher's seminal work on univariate $k$-statistics in 1930 \cite{RAF-30} and 
	Wishart and Kendall's later work extending them to the multivariate case (e.g., 
	\cite{MGK-40}). These papers provide formulas for low-order $k$-statistics, and they 
	describe the process of 
	symbol manipulation that can be used to construct higher-order formulas, but they 
	stop short of an explicit, general expression for $k$-statistics. The objective of 
	this note is to provide such an expression, as well as a self-contained derivation. 
	
	The technical content of this paper is not really novel. We highlight the deep 
	connection that cumulants and $k$-statistics have with M\"obius inversion on the 
	partition lattice, but this connection has been known since at least 1983 
	\cite{TPS-83} and has been noted in subsequent work \cite{GCR-JS:00, PM-08}. The same 
	formulas for multivariate $k$-statistics derived here have also been constructed 
	through the umbral calculus \cite{EDN-GG-DS:08, EDN-GG-DS:09}, though 
	interpreting these formulas requires a nontrivial investment of effort into learning 
	the umbral calculus formalism. Software packages are also available for 
	$k$-statistics. \textit{MathStatica}, a proprietary \textit{Mathematica} package, 
	provides methods for symbolic $k$-statistic formulas \cite{CR-MDS:02}. An R package, 
	\textit{kStatistics} \cite{EDN-GG:19}, is available to compute multivariate 
	$k$-statistics of data samples. Our own library \textit{PyMoments} implements 
	multivariate $k$-statistics in Python \cite{KDS:20}. Thus, rather than 
	providing new insight into the problem of cumulant estimation, the aim of this paper 
	is to serve as a quick, accessible, and explicit reference for multivariate 
	$k$-statistics. 
	
	We also hope this paper will invite discussion regarding the efficient computation of 
	$k$-statistics. The time complexity of computing these statistics scales poorly with 
	order, so in order to make higher-order $k$-statistics useful in real-world 
	applications, it is necessary to optimize their efficiency. We briefly discuss some 
	possible avenues toward efficient computation or approximation of these statistics, 
	but this topic is still under-explored.
	
	The paper is organized as follows. In the remainder of this section, we introduce the 
	preliminary mathematical concepts that are needed to derive $k$-statistics, including 
	the definition of cumulants themselves (Section \ref{def:cum}) and their connection 
	to M\"obius inversion on the partition lattice (Section \ref{sect:mobius}). Section 
	\ref{sect:kstat} contains the main results---a definition of and explicit formula for 
	$k$-statistics in terms of raw sample moments (Definition \ref{def:kstat}), and a 
	proof of their unbiased estimation that reveals a derivation of these statistics 
	(Section \ref{sect:proof}). The last part of the paper, Section 
	\ref{sect:computation}, briefly discusses some points on the computational efficiency 
	of evaluating $k$-statistics. 
		
	\subsection{Preliminaries}
		
	\paragraph*{Multisets and Multi-indices}
	A \textit{multiset} is a set that allows for repeated elements. There are two ways to 
	represent a multiset. The simplest representation is explicit enumeration of the 
	elements, e.g., $[x_1, x_2, \dots, x_n]$, where it is possible that $x_i = x_j$. When 
	the universe of possible elements in the multiset is clear from context, another 
	representation is to use a \textit{multi-index}, which assigns an integer 
	multiplicity to every element in the universe. For example, when we are considering 
	multisets with elements drawn from $\{1, 2, \dots, n\}$, we can encode the multiset 
	using a multi-index $\alpha: \{1, 2, \dots, n\} \to \mathbb Z_{\ge 0}$, where 
	$\alpha(i)$ is the multiplicity of $i$ in the multiset. We will often use ``multiset 
	generator'' notation to describe a multiset; or example, $[i ~\text{mod}~ 2 \mid i 
	\in \{1, 2, 3, 4, 5\}] = [0, 0, 1, 1, 1]$.
	
	\begin{figure} 
		\centering
		\includegraphics[width=0.4\linewidth]{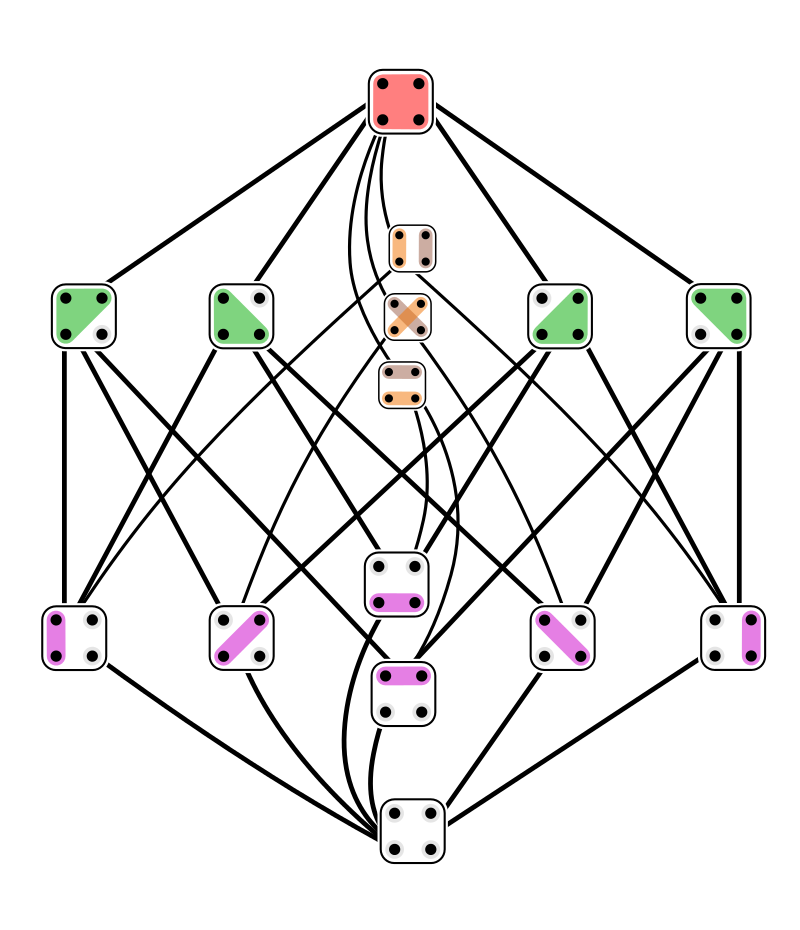}
		\caption{Visual representation of the partition lattice of a 4-element set. Edges 
		in the diagram indicate that the lower partition refines the upper partition. 
		This image is credited to Tilman Piesk and distributed under a CC BY 3.0 license.}
		\label{fig:partitions}
	\end{figure}
	
	\paragraph*{Partitions}
	A \textit{partition} of a set $S$ is a collection of mutually disjoint subsets $B_1, 
	B_2, \dots, B_k \subseteq S$, called \textit{blocks}, such that $\bigcup_{i = 1}^k 
	B_i = S$. We denote partitions as sets of blocks, e.g., $\pi = \{B_1, B_2, \dots, 
	B_k\}$. The \textit{size} of a partition is the number of blocks: $|\pi| = k$. Given 
	two partitions $\pi, \rho$ of the same set, we say that $\pi$ \textit{refines} 
	$\rho$, and write $\pi \le \rho$, if every block in $\pi$ is the subset of a block in 
	$\rho$. The \textit{partition lattice} $(\Pi_n, \le)$ is the poset of partitions of 
	the set $\{1, 2, \dots, n\}$, with refinement as a partial order. Within the 
	partition lattice, note that the unique partition with one block is the unique 
	maximum element; similarly, the unique partition with $n$ blocks is the unique 
	minimum element. Figure \ref{fig:partitions} provides a visual representation of the 
	partition lattice $\Pi_4$.
	
	The number of partitions on $\Pi_n$ with size $|\pi| = k$ is given by 
	\textit{Stirling's number of the second kind}, and is given by
	\[
		\stirling{n}{k} = \frac{1}{k!} \sum_{i=0}^{k-1} (-1)^i \binom k i (k - i)^n
	\]
	for positive $n$. The total number of partitions in $\Pi_n$ is known as 
	\textit{Bell's number}:
	\[
		B_n = |\Pi_n| = \sum_{k=1}^n \stirling n k
	\]
	
	\paragraph*{General Notation}
	Given two non-negative integers $k \le n$, the \textit{falling factorial} is the 
	quantity $(n)_k = n (n - 1) \cdots (n - k + 1)$.
%
%
	
	\subsection{Cumulants} \label{def:cum}
	We begin with a formal definition of cumulants and (implicitly) a review of our 
	notational conventions. Given a random vector $X = 
	\begin{pmatrix} X_1 & X_2 & \cdots & X_n \end{pmatrix}$, where $X_i$ are scalar 
	random variables, define the \textit{moment generating 
	function} $M_X: \mathbb R^n \to \mathbb R$ and the \textit{cumulant generating 
	function} $K_X: \mathbb R^n \to \mathbb R$ by
	\[
		M_X(t) = \E\left[ e^{t^\top X} \right], \qquad K_X(t) = \log M_X(t)
	\]
	Assuming that $M_X(t)$ and $K_X(t)$ admit Taylor expansions about $t = \vect 0_n$, we 
	can write
	\begin{align*}
		M_X(t) &= 1 + \sum_{i=1}^n m_{[i]} t_i + \frac{1}{2} \sum_{i, j = 1}^n m_{[i,j]} 
		t_i t_j + \sum_{i, j, k = 1}^n m_{[i,j,k]} t_i t_j t_k + \cdots \\
		K_X(t) &= \sum_{i=1}^n \kappa_{[i]} t_i + \frac 1 2 \sum_{i, j = 1}^n 
		\kappa_{[i,j]} 
		t_i t_j + \sum_{i, j, k = 1}^n \kappa_{[i,j,k]} t_i t_j t_k + \cdots
	\end{align*}
	where the coefficients $m_\alpha$ and $\kappa_\alpha$ are defined for any multiset 
	from the indices $\{1, 2, \dots, n\}$. The coefficients in the expansion of the 
	moment generating function are familiar---for example, 
	\[
		m_{[i,j,k]} = \left. \frac{\partial^3M_X(t)}{\partial t_i \partial t_j \partial 
		t_k} \right|_{t = \vect 0_n}
		= \left. \E \left[
			\frac{\partial^3}{\partial t_i \partial t_j \partial t_k} e^{t^\top X}
		\right] \right|_{t = \vect 0_n}
		= \E \left[X_i X_j X_k\right]
	\]
	is a third-order raw moment. Of course, this relationship holds true in general: the 
	coefficients $m_\alpha$ in the expansion of the moment generating function are 
	precisely the raw moments of $X$.
	
	Cumulants are defined similarly, as coefficients in the Taylor expansion of 
	$K_X(t)$. Formally, given any multiset from $\{1, 2, \dots, n\}$, or equivalently, 
	given any multi-index $\alpha$ on $\{1, 2, \dots, n\}$, we define the cumulant
	\begin{equation}
		\kappa_\alpha(X) = \left. \frac{\partial^{|\alpha|} 
		K_X(t)}{\partial_{t_1}^{\alpha(1)} 
		\partial_{t_2}^{\alpha(2)} \cdots \partial_{t_n}^{\alpha(n)}} \right|_{t =\vect 
		0_n}
		\label{eq:cum}
	\end{equation}
	as the coefficient of the term $\frac{1}{|\alpha|!} t_1^{\alpha(1)} t_2^{\alpha(2)} 
	\cdots t_n^{\alpha(n)}$ in the series expansion of $K_X(t)$. The \textit{order} of a 
	cumulant is the size $|\alpha|$. Low-order cumulants have 
	familiar interpretations, as the next few examples demonstrate:
	
	\begin{example}[First-Order Cumulants] \label{ex:cum-first}
		Consider a single random variable $X_i$. The first-order cumulant of this 
		variable is
		\[
			\kappa_{[i]}(X) = \left. \frac{\partial K_X(t)}{\partial_{t_i}} \right|_{t = 
			\vect 0_n}
			= \left. \frac{1}{M_X(t)} \frac{\partial M_X(t)}{\partial t_i} \right|_{t = 
			\vect 0_n}
			= m_{[i]} = \E[X_i]
		\]
		Thus, first-order cumulants are identical to first-order raw moments, i.e., means.
		\oprocend
	\end{example}

	\begin{example}[Second-Order Cumulants] \label{ex:cum-second}
		Consider a pair of random variables $X_i, X_j$, possibly repeating. The 
		second-order cumulant of this pair of variables is
		\begin{align*}
			\kappa_{[i,j]}(X) = \left.
				\frac{\partial^2 K_X(t)}{\partial t_i \partial t_j}
			\right|_{t = \vect 0_n}
			&= \left. \frac{1}{M_X(t)} \frac{\partial^2 M_X(t)}{\partial t_i \partial 
			t_j} - \frac{1}{M_X(t)^2} \frac{\partial M_X(t)}{\partial t_i} \frac{\partial 
			M_X(t)}{\partial t_j} \right|_{t = \vect 0_n} \\
			&= m_{[i, j]} - m_{[i]} m_{[j]} \\
			&= \E[X_i X_j] - \E[X_i] \E[X_j] \\
			&= \text{cov}(X_i, X_j)
		\end{align*}
		Thus, second-order cumulants are identical to covariances. 
		\oprocend
	\end{example}

	\begin{example}[Third-Order Cumulants] \label{ex:cum-third}
		Consider a triple of random variables $X_i, X_j, X_k$, possibly repeating. After 
		evaluating and simplifying the appropriate third derivative, we find that
		\begin{align*}
			\kappa_{[i, j, k]}(X) = \left.
			\frac{\partial^3 K_X(t)}{\partial t_i \partial t_j \partial t_k}
			\right|_{t = \vect 0_n}
			&= 2 m_{[i]} m_{[j]} m_{[k]} - m_{[i]} m_{[i,j]} - m_{[j]} m_{[i,k]} - 
			m_{[k]} m_{[i,j]} + m_{[i,j,k]} 
		\end{align*}
		In particular, if $X_i = X_j = X_k$, we obtain the third \textit{univariate} 
		cumulant of the random variable $X_i$:
		\begin{align*}
			\kappa_{[i, i, i]}(X) &= 2 \E[X_i]^3 - 3 \E[X_i] \E[X_i]^2 + \E[X_i^3] \\
			&= \E[(X_i - \E[X_i])^3] \\
			&= \text{var}(X_i)^{3/2} \text{skew}(X_i)
		\end{align*}
		where $\text{skew}(\cdot)$ is the moment coefficient of skewness. In other words, 
		third univariate cumulants are identical to third central moments, thereby 
		quantifying	the skewness of a distribution.
	\end{example}

	\subsection{Cumulants, Raw Moments, and M\"obius Inversion}
	\label{sect:mobius}
	
	As the previous three examples suggest, cumulants and raw moments are closely 
	related---after all, cumulants and raw moments are the series coefficients of
	functions related by a log transform. High-order derivatives of the logarithm result 
	in an abundance of terms, which rapidly get out of hand when trying to derive 
	high-order cumulants manually. Fortunately, by invoking the multivariate version of 
	Fa\`a di Bruno's formula, we can use some powerful ideas from combinatorics to 
	compactly represent these computations. This subsection will reveal a general formula 
	for expression cumulants in terms of raw moments, which will prove useful in our 
	later derivation of $k$-statistics.
	
	It is a bit simpler to start in the reverse direction, writing raw moments in terms 
	of cumulants, and then using M\"obius inversion to obtain our desired expressions. 
	Writing raw moments in terms of cumulants is really a straightforward application of 
	Fa\`a di Bruno's generalization of the chain rule:
	
	\begin{lemma}[Fa\`a di Bruno's Formula]
		Let $n$ and $k$ be positive integers, and let $f: \mathbb R \to \mathbb R$ and 
		$g: \mathbb R^k \to \mathbb R$ be a pair of functions that are differentiable to 
		order $n$. Then
		\begin{equation}
			\frac{\partial^n f(g(x))}{\partial x_1 \partial x_2 \cdots \partial x_n} 
			= \sum_{\pi \in \Pi_n} f^{(|\pi|)}(g(x)) \prod_{B \in \pi} 
			\frac{\partial^{|B|} g(x)}{\prod_{i \in B} \partial x_i}
			\label{eq:fdb}
		\end{equation} 
		where $f^{(j)}(\cdot)$ denotes the $j$th derivative, $\pi$ loops through every 
		partition of the set $\{1, 2, \dots, n\}$, $B$ 
		loops through each block in a given partition, and $i$ loops through each element 
		of $\{1, 2, \dots, n\}$ contained within a given block.
	\end{lemma}
	\noindent
	See \cite{MH:06} for an overview and proof of this formula (as well as a different 
	discussion of the application we are about to see). Equation \eqref{eq:fdb} is an 
	entrypoint that will allow us to apply the combinatorics of the partition lattice to 
	the study of cumulants.
	
	In order to write raw moments in terms of cumulants, we can express $M_X(t) = 
	f(g(t))$, where $f(\cdot) = \exp(\cdot)$ and $g(t) = K_X(t)$. Therefore, applying 
	\eqref{eq:fdb}, we obtain
	\begin{align*}
		m_{[i_1, i_2, \dots, i_k]} &= \left.\frac{\partial^n 
		e^{K_X(t)}}{\partial t_{i_1} \partial t_{i_2} \cdots \partial t_{i_k}} \right|_{t 
		= \vect 0_n} \\
		&= \sum_{\pi \in \Pi_k} \left. \frac{d \exp(y)}{d y} \right|_{y = K_X(\vect 0_n)}
		\left( \prod_{B \in \pi}  \left. \frac{\partial^{|B|} K_X(t)}{\prod_{j \in B} 
		\partial t_{i_j}} \right|_{t =\vect 0_n} \right) \\
		&= \sum_{\pi \in \Pi_k} \prod_{B \in \pi} \left. \frac{\partial^{|B|} 
		K_X(t)}{\prod_{i \in B} \partial t_{i_j}} \right|_{t =\vect 0_n} 
	\end{align*}
	We recognize the inner derivative as a cumulant of the random vector $X = 
	\begin{pmatrix} X_1 & X_2 & \cdots & X_n \end{pmatrix}$, in particular, the cumulant 
	corresponding to the multiset of indices $[i_j \mid j \in B]$. Therefore, simplifying 
	this equation, we can express the raw moment in terms of cumulants, as follows:
	\begin{equation}
		m_{[i_1, i_2, \dots, i_k]} = \sum_{\pi \in \Pi_k} \prod_{B \in \pi} \kappa_{[i_j 
		\mid j \in B]}
		\label{eq:mom-as-cumulant}
	\end{equation} 
	
	Let us consider some small examples of this formula:
	
	\begin{example}[Low-Order Raw Moments]
		Let us consider how to construct fist-order and second-order raw moments from 
		cumulants. Given any particular random variable $X_{i}$ from the vector $X = 
		\begin{pmatrix} X_1 & X_2 & \cdots & X_n \end{pmatrix}$, we know that $m_{[i]} = 
		\kappa_{[i]}$, either by applying \eqref{eq:mom-as-cumulant} or recalling Example 
		\ref{ex:cum-first}. Next, given a pair of (possibly identical) random variables 
		$X_{i_1}, X_{i_2}$ from the vector, we use \eqref{eq:mom-as-cumulant} to compute
		\begin{align*}
			m_{[i_1, i_2]} = \sum_{\pi \in \Pi_2} \prod_{B \in \pi} \kappa_{[i_j \mid j 
			\in B]}
			&= \prod_{B \in \{ \{1\}, \{2\} \}} \kappa_{[i_j \mid j 
				\in B]} + \prod_{B \in \{ \{1, 2\} \}} \kappa_{[i_j \mid j 
				\in B]} \\
			&= \kappa_{[i_1]} \kappa_{[i_2]} + \kappa_{[i_1, i_2]}
		\end{align*} 
		Combining these two results, we see that
		\[
			\kappa_{[i_1, i_2]} = m_{[i_1, i_2]} - \kappa_{[i_1]} \kappa_{[i_2]}
			= m_{[i_1, i_2]} - m_{[i_1]} m_{[i_2]} = \text{cov}(X_{i_1}, X_{i_2})
		\]
		replicating our conclusion from Example \ref{ex:cum-second}.
		\oprocend
	\end{example}
	
	This example hints at the possibility of inverting \eqref{eq:mom-as-cumulant} to 
	express cumulants as functions of central moments. It turns out that stating a 
	general formula for this inversion is straightforward, thanks to M\"obius inversion. 
	The general topic of M\"obius inversion on posets is beyond the 
	scope of this note, but many good lecture notes are available online for an easy 
	introduction (for example, \cite{PB:15}), or the reader may refer to a text like 
	Aigner \cite{MA:97} for a more detailed and rigorous discussion. Fortunately, 
	M\"obius inversion on the partition lattice is a standard example in this area of 
	combinatorics, so we can cut to the chase and state the needed result:
	
	\begin{lemma}[M\"obius Inversion on the Partition Lattice: Part I] 
	\label{lem:mobius-1}
		Consider two functions $f, g: \Pi_n \to \mathbb R$. The following are equivalent:
		\begin{enumerate}
			\item
			\begin{equation}
				f(\pi) = \sum_{\rho \le \pi} g(\rho), \qquad \forall \pi \in \Pi_n
				\label{eq:mi-1-f}
			\end{equation}
			\item
			\begin{equation}
				g(\pi) = \sum_{\rho \le \pi} (-1)^{|\rho| - 1} (|\rho| - 1)! f(\rho), 
				\qquad \forall \pi \in \Pi_n
				\label{eq:mi-1-g}
			\end{equation}
		\end{enumerate}
	\end{lemma}

	Lemma \ref{lem:mobius-1} provides a handy formula to invert sums over refinements of 
	a given element of the partition lattice. While not immediately obvious, 
	Lemma \ref{lem:mobius-1} can be used to invert \eqref{eq:mom-as-cumulant}, leading to 
	the following result:
	\begin{theorem}[Cumulants from Raw Moments]
		Consider the raw moments $m$ and cumulants $\kappa$ of a random vector $X = 
		\begin{pmatrix} X_1 & X_2 & \cdots & X_n \end{pmatrix}$. For any multiset $[i_1, 
		i_2, \dots, i_k]$ from the indices $\{1, 2, \dots, n\}$, the corresponding 
		cumulant can be expressed in terms of the raw moments by
		\begin{equation}
			\kappa_{[i_1, i_2, \dots, i_k]} = \sum_{\pi \in \Pi_k} (-1)^{|\pi| - 1} 
			(|\pi| - 1)! \prod_{B \in \pi} m_{[i_j \mid j \in B]}
			\label{eq:cum-as-mom}
		\end{equation}
		\label{thm:cum-as-mom}
	\end{theorem}

	\begin{proof}
		Define two maps $f_k, g_k: \Pi_k \to \mathbb R$ by
		\[
			g_k(\pi) = \prod_{B \in \pi} \kappa_{[i_j \mid j \in B]} 
		\]
		and
		\[
			f_k(\pi) = \prod_{B \in \pi} m_{[i_j \mid j \in B]} 
		\]
		Let $\hat 1_k = \{ \{1, 2, \dots, k\}\}$ be the maximum partition in $\Pi_k$. 
		Then \eqref{eq:mom-as-cumulant} can be re-written $f_k(\hat 1_k) = \sum_{\rho \le 
		\hat 1_k} g_k(\rho)$, since the set $\Pi_k$ is precisely the set of refinements 
		of $\hat 1_k$. In fact, this is enough to conclude that $f_k(\pi) = \sum_{\rho 
		\le \pi} g_k(\rho)$ for all $\pi \in \Pi_k$. This is because we can write
		\[
			f_k(\pi) = \prod_{B \in \pi} f_B(\hat 1_B)
		\]
		where $\hat 1_B$ is the maximum element of the lattice of partitions of $B$, and 
		$f_B$ is defined similar to $f_k$, but on the elements of $B$ instead of $\{1, 
		2, \dots, k\}$. Invoking \eqref{eq:mom-as-cumulant}, we have that $f_B(\hat 1_B) 
		= \sum_{\rho_B \le \hat 1_B} g_B(\rho_B)$ (where $g_B$ is defined similar to 
		$g_k$), so that
		\begin{align*}
			f_k(\pi) = \prod_{B \in \pi} \sum_{\rho_B \le \hat 1_B} g_B(\rho_B)
			&= \sum_{\rho_{B_1} \le \hat 1_{B_1}} \sum_{\rho_{B_2} \le \hat 1_{B_2}}
			\cdots \sum_{\rho_{B_{|\pi|}} \le \hat 1_{B_{|\pi|}}} \prod_{B \in \pi} 
			g_B(\rho_B) \\
			&= \sum_{\rho \le \pi} g_k(\rho)
		\end{align*}
		since the set of all tuples of refinements $(\rho_{B_1}, \rho_{B_2}, \dots, 
		\rho_{B_{|\pi|}})$ is isomorphic to the product of lattices $\Pi_{|B_1|} 
		\Pi_{|B_2|} \cdots \Pi_{|B_{|\pi|}}$, which is itself isomorphic to the set of 
		all refinements of $\pi$. Thus $f_k$ and $g_k$ satisfy \eqref{eq:mi-1-f}, so we 
		invoke the M\"obius inversion in Lemma \ref{lem:mobius-1} to obtain 
		\eqref{eq:mi-1-g}. In particular, evaluating \eqref{eq:mi-1-g} on $\hat 1_k$, we 
		obtain
		\begin{align*}
			g_k(\hat 1_k) &= \sum_{\rho \le \hat 1_k} (-1)^{|\rho| - 1} (|\rho| - 1)! 
			f_k(\rho) \\
			&= \sum_{\rho \in \Pi_k} (-1)^{|\rho| - 1} (|\rho| - 1)! \prod_{B \in \pi} 
			m_{[i_j \mid j \in B]}
		\end{align*}
		But $g_k(\hat 1_k) = \kappa_{[i_1, i_2, \dots, i_k]}$, so we obtain 
		\eqref{eq:cum-as-mom}.
	\end{proof}

	\section{Multivariate $k$-Statistics} \label{sect:kstat}
	
	Now that we have examined multivariate cumulants, our next challenge is to estimate 
	them from a sample. Once again, suppose that we have a random vector $X = 
	\begin{pmatrix} X_1 & X_2 & \cdots & X_n \end{pmatrix}^\top$, distributed according 
	to some joint distribution $F$. Further suppose that, instead of knowing $F$, all we 
	have is an i.i.d. sample $x_1, x_2, \dots, x_N$ from this distribution, 
	where $x_t \in \mathbb R^n$. We would like to estimate the cumulants of $X$ using 
	some statistic, i.e., some function of the data $x_1, x_2, \dots, x_N$.
	
	It turns out that we can obtain an unbiased estimate of each cumulant using 
	\textit{raw sample moments}. For each multiset $[i_1, i_2, \dots, i_k]$ of the 
	indices $\{1, 2, \dots, n\}$, the corresponding raw sample moment is the 
	statistic given by
	\begin{equation}
		\hat m_{[i_1, i_2, \dots, i_k]} = \frac{1}{N} \sum_{t = 1}^N x_{t, i_1} x_{t, 
		i_2} 
		\cdots x_{t, i_k}
		\label{eq:samp-moment}
	\end{equation}
	Because the observations in the sample are independent, $\hat m_{[i_1, i_2, \dots, 
	i_k]}$ is an unbiased estimator for the raw moment $m_{[i_1, i_2, \dots, i_k]}$. 
	Furthermore, we can use raw sample moments to obtain an unbiased estimates of 
	cumulants:
	\begin{definition}[$k$-Statistic] \label{def:kstat}
		Consider the random vector $X = \begin{pmatrix} X_1 & X_2 & \cdots & X_n 
		\end{pmatrix}^\top$ and some multiset $[i_1, i_2, \dots, i_k]$ from the indices 
		$\{1, 2, \dots, n\}$. Given a sample of $X$ with at least $N \ge k$ observations, 
		the corresponding \textit{$k$-statistic} is given by
		\begin{equation}
			k_{[i_1, i_2, \dots, i_k]} = \sum_{\pi \in \Pi_k} (-1)^{|\pi| - 1} c_\pi 
			\prod_{B \in \pi} \hat m_{[i_j \mid j \in B]}
			\label{eq:k-stat}
		\end{equation}
		where we define a positive coefficient for each partition in $\Pi_k$ by
		\begin{equation}
			c_\pi = N^{|\pi|} \sum_{b_1 = 1}^{|B_1|} \sum_{b_2 = 1}^{|B_2|} \cdots 
			\sum_{b_{|\pi|} = 1}^{|B_{|\pi|}|}  
				\frac{\left( \sum_{j=1}^{|\pi|} b_j - 
				1\right)!}{(N)_{\sum_{j=1}^{|\pi|} b_j}}
			\left(
			\prod_{j = 1}^{|\pi|} \stirling{|B_j|}{b_j} (b_j - 1)!
			\right) , \qquad \forall \pi \in \Pi_k
			\label{eq:coef}
		\end{equation} 
		and $B_1, B_2, \dots, B_{|\pi|}$ are the blocks of the partition $\pi$.
		\oprocend
	\end{definition} 

	Equation \eqref{eq:k-stat} is a linear combination of products of raw sample moments 
	(which are computed from the data). The linear combination itself involves a sum over 
	the partition lattice with coefficients of alternating sign, which hints at M\"obius 
	inversion. Indeed, we can derive the $k$-statistic by using M\"obius inversion to 
	correct the bias of products of raw sample moments. We provide a much more detailed 
	derivation in the next section, when we prove that $k$-statistics are unbiased:

	\begin{theorem}[$k$-Statistics are Unbiased Estimators of Cumulants]
		\label{thm:k}
		Consider the random vector $X = \begin{pmatrix} X_1 & X_2 & \cdots & X_n 
		\end{pmatrix}^\top$, and let $[i_1, i_2, \dots, i_k]$ be any multiset of the 
		indices $\{1, 2, \dots, n\}$. The $k$-statistic computed from any i.i.d. sample 
		of $X$ with at least $k$ observations is an unbiased estimator of the cumulant, 
		i.e.,  
		\[
			\E\left[
				k_{[i_1, i_2, \dots, i_k]}
			\right] = \kappa_{[i_1, i_2, \dots, i_k]}
		\] 
	\end{theorem}  
	First, we will consider some lower-order examples of multivariate $k$-statistics.
	
	\begin{example}[First-Order $k$-Statistics]
		The partition lattice $\Pi_1$ consists of only one element $\{ \{1\}\}$, so 
		first-order $k$-statistics are easily computed as
		\[
			k_{[i]} = c_{\{\{1\}\}} \hat m_{[i]} = \hat m_{[i]}
		\]
		Thus, as expected from Example \ref{ex:cum-first}, first-order $k$-statistics are 
		merely sample means. \oprocend
	\end{example}

	\begin{example}[Second-Order $k$-Statistics]
		The partition lattice $\Pi_2$ consists of two elements: $\{ \{1, 2\} \}$, and $\{ 
		\{1\}, \{2\}\}$. The corresponding coefficients are
		\[
			c_{12} = N^1 \left( \frac{(0)!}{(N)_1} \stirling{2}{1}(0)! +  
			\frac{(1)!}{(N)_2} \stirling{2}{2} (1)! \right) 
			= 1 + \frac{1}{N - 1}
			= \frac{N}{N - 1}
		\]
		and
		\[
			c_{1|2} = N^2 \frac{(1)!}{(N)_2} \stirling{1}{1}\stirling{1}{1}(0)!(0)!
			= \frac{N^2}{N(N - 1)} = \frac{N}{N - 1}
		\]
		Therefore
		\[
			k_{[i_1 i_2]} = \frac{N}{N - 1} \left( \hat m_{[i_1, i_2]} +  \hat 
			m_{[i_1]} \hat m_{[i_2]} \right)
		\]
		which we recognize as the classical unbiased estimator for covariance. Of course, 
		this is exactly what we should expect after Example \ref{ex:cum-second}. \oprocend
	\end{example}

	\begin{example}[Third-Order $k$-Statistics]
		The partition lattice $\Pi_3$ has five elements: $\{\{1, 2, 3\}\}$,\; $\{ \{1\}, 
		\{2, 3\}\}$,\; $\{ \{2\}, \{1, 3\}\}$,\; $\{ \{3\}, \{1, 2\} \}$,\; and $\{\{1\}, 
		\{2\}, \{3\}\}$. We first compute the respective coefficients $c_{123}$, 
		$c_{1|23}$, $c_{2|13}$, $c_{3|12}$, and $c_{1|2|3}$:
		\begin{align*}
			c_{123} &= N^1 \left(
				\frac{(0)!}{(N)_1} \stirling{3}{1} (0)! + \frac{(1)!}{(N)_2} 
				\stirling{3}{2} (1)! + \frac{(2)!}{(N)_3} \stirling{3}{3} (2)!
			\right)
			= 1 + \frac{3}{N - 1} + \frac{4}{(N - 1)(N - 2)} \\
			c_{1|23} &= N^2 \left(
				\frac{(1)!}{(N)_2} \stirling{1}{1} \stirling{2}{1} (0)! (0)! +
				\frac{(2)!}{(N)_3} \stirling{1}{1} \stirling{2}{2} (0)! (1)!
			\right)
			= \frac{N}{N - 1} + \frac{2 N}{(N - 1)(N-2)} \\
			c_{1|2|3} &= N^3 \frac{(2)!}{(N)_3} \stirling{1}{1} \stirling{1}{1} 
			\stirling{1}{1} (0)! (0)! (0)! = \frac{2N^2}{(N - 1)(N - 2)}
		\end{align*}
		Note that $c_\pi$ depends only on the number and size of each block, and not the 
		blocks themselves, so $c_{1|23} = c_{2|13} = c_{3|12}$. Substituting these 
		coefficients into \eqref{eq:k-stat} and simplifying, we obtain
		\begin{align*}
			k_{[i_1, i_2, i_3]} &= c_{123} \hat m_{[i_1, i_2, i_3]} + c_{1|23} \left( 
			\hat m_{[i_1]} \hat m_{[i_2, i_3]} + \hat m_{[i_2]} \hat m_{[i_1, i_3]} + 
			\hat m_{[i_3]} \hat m_{[i_1, i_2]} \right) + c_{1|2|3} \hat m_{[i_1]} \hat 
			m_{[i_2]} \hat m_{[i_3]} \\
			&= \frac{N^2}{(N-1)(N-2)} \left(
				\hat m_{[i_1, i_2, i_3]}  -  
				\hat m_{[i_1]} \hat m_{[i_2, i_3]} - \hat m_{[i_2]} \hat m_{[i_1, 
					i_3]} - \hat m_{[i_3]} \hat m_{[i_1, i_2]}
				+ 2 \hat m_{[i_1]} \hat m_{[i_2]} \hat m_{[i_3]}
			\right) 
		\end{align*}
		In particular, if $X_{i_1} = X_{i_2} = X_{i_3}$, we obtain the third univariate 
		$k$-statistic
		\[
			k_{[i, i, i]} = \frac{N^2}{(N - 1)(N - 2)} \left(
				\hat m_{[i, i, i]} - 3 \hat m_{[i]} \hat m_{[i, i]} + \hat m_{[i]}^3
			\right)
		\]
		\oprocend
	\end{example}
	
	\subsection{Proof of Theorem \ref{thm:k}}
	\label{sect:proof}
	
	The general outline of the proof is as follows. We will first note that, while raw 
	sample moments are unbiased estimators of raw moments, it is still the case that 
	\textit{products} of raw sample moments provide biased estimates for products of raw 
	moments. The first step will be to quantify this bias. Second, we will once again use 
	M\"obius inversion over the partition lattice to obtain an \textit{unbiased} 
	estimator of products of raw moments, in terms of products of raw sample moments. 
	Finally, we will substitute this estimator into \eqref{eq:cum-as-mom} and simplify.
	
	We begin by examining the expected value of raw sample moments: 
	
	\begin{lemma}[Bias of Products of Raw Sample Moments] \label{lem:bias}
		Consider the random vector $X = \begin{pmatrix} X_1 & X_2 & \cdots & X_n 
		\end{pmatrix}^\top$, and consider some multiset $[i_1, i_2, \dots, i_k]$ from the 
		indices $\{1, 2, \dots, n\}$. For every $\pi \in \Pi_k$, we have
		\begin{equation}
			\E\left[ \prod_{B \in \pi} \hat m_{[i_j \mid j \in B]} \right]
			= \frac{1}{N^{|\pi|}}\sum_{\rho \ge \pi} (N)_{|\rho|} \prod_{C \in \rho} 
			m_{[i_j \mid j \in C]}
			\label{eq:bias}
		\end{equation}
	\end{lemma}

	\begin{proof}
		Our first step is to switch the order of sums and products, as follows:
		\begin{align*}
			\E\left[ \prod_{B \in \pi} \hat m_{[i_j \mid j \in B]} \right]
			&= \frac{1}{N^{|\pi|}}\E\left[
				\prod_{B \in \pi} \sum_{t=1}^N \prod_{j \in B} x_{t,i_j}
			\right] \\
			&= \frac{1}{N^{|\pi|}} \sum_{t_1 = 1}^N \sum_{t_2 = 1}^N \cdots 
			\sum_{t_{|\pi|} = 1}^N \E \left[
				\prod_{j = 1}^{|\pi|} \prod_{\ell \in B_j} x_{t_j, i_\ell}
			\right]
		\end{align*}
		Note that the inner expectation depends on which of the observations 
		$t_1, t_2, \dots, t_{|\pi|}$ are identical, since observations at distinct times 
		are independent, allowing us to factor the expected value. With this in mind, we 
		will partition the hypercube of $t$-indices that we are summing over 
		into equivalence classes, based on partitions of the set $\{1, 2, \dots, 
		|\pi|\}$. Given such a partition $\sigma = \{C_1, C_2, \dots, C_{|\sigma|}\}$, we 
		define the equivalence class $[\sigma]$ as the set of index tuples $(t_1, t_2, 
		\dots, t_{|\pi|})$ with the following property: for all $i, j \in \{1, 2, \dots, 
		|\pi|\}$, we have that $t_i = t_j$ if and only if $t_i, t_j \in C$ for some block 
		$C \in \sigma$. In other words, the blocks of $\sigma$ represent elements of the 
		$t$-index that are identical. Clearly each of the $N^{|\pi|}$ index tuples in the 
		sum belong to some equivalence class $[\sigma]$. Furthermore, each equivalence 
		class contains $(N)_{|\sigma|}$ elements: $N$ possible values for indices in the 
		first block, $N - 1$ possible values in the second block, and so on.
		
		For all $(t_1, t_2, \dots, t_{|\pi|}) \in [\sigma]$, we have the following 
		property:
		\[
			\E \left[
			\prod_{j = 1}^{|\pi|} \prod_{\ell \in B_j} x_{t_j, i_\ell}
			\right]
			= \prod_{C \in \sigma} \E\left[ \prod_{j \in C} \prod_{\ell \in B_j} 
			X_{i_\ell} 
			\right]
			= \prod_{C \in \sigma} \E \left[
				\prod_{\ell \in \bigcup_{j \in C} B_j} X_{i_\ell}
			\right]
			= \prod_{C \in \sigma} m_{[i_\ell \mid \ell \in \bigcup_{j \in C} B_j]}
		\]
		This follows because the expected value factors along the blocks of $\sigma$, 
		since the blocks have pairwise-distinct times, and thus the observations in each 
		block are pairwise independent. Then we can write
		\[
			\E\left[ \prod_{B \in \pi} \hat m_{[i_j \mid j \in B]} \right]
			 = \frac{1}{N^{|\pi|}} \sum_{\sigma 
			\in \Pi_{|\pi|}} (N)_{|\sigma|} \prod_{C \in 
				\sigma} m_{[i_\ell \mid \ell \in \bigcup_{j \in C} B_j]}
		\] 
		The final step is to note that there is a bijection between partitions of $\{1, 
		2, \dots |\pi|\}$ and partitions of $\{1, 2, \dots, k\}$ that are coarser than 
		$\pi$. This bijection is easy to see: for each block $C \in \sigma$, replace all 
		of the blocks in $\pi$ with coarser blocks $\bigcup_{j \in C} B_j$, resulting in 
		a coarser partition $\rho \ge \pi$. Due to this bijection, we can re-write the 
		sum over $\sigma$ as a sum over coarser partitions $\rho \ge \pi$, obtaining 
		\eqref{eq:bias}.
	\end{proof} 

	Equation \eqref{eq:bias} has a somewhat familiar form---a sum over partitions that 
	are coarser than $\pi$. Recall from the proof of Theorem \ref{thm:cum-as-mom} that we 
	used M\"obius inversion to invert a sum over partitions that \textit{refine} $\pi$. 
	While the direction of the sum makes a difference---we cannot use Lemma 
	\ref{lem:mobius-1} in this case---the partition lattice still admits a M\"obius 
	inversion formula to invert \eqref{eq:bias}. Once again, we will state the needed 
	formula here, and direct the interested reader to a combinatorics text like 
	\cite{MA:97}:

	\begin{lemma}[M\"obius Inversion on the Partition Lattice: Part II] 
	\label{lem:mobius-2}
		Consider two functions $f, g: \Pi_n \to \mathbb R$. For two partitions $\pi \le 
		\rho \in \Pi_n$, let $\Sigma(\pi, \rho)$ denote the set of $|\rho|$ partitions 
		that, when applied to the blocks of $\rho$, yield the refinement $\pi$. Then the 
		following are equivalent:
		\begin{enumerate}
			\item
			\begin{equation}
				f(\pi) = \sum_{\rho \ge \pi} g(\rho), \qquad \forall \pi \in \Pi_n
				\label{eq:mobius-2-f}
			\end{equation}
			\item
			\begin{equation}
				g(\pi) = \sum_{\rho \ge \pi} \left( (-1)^{|\pi| - |\rho|} 
					\prod_{\sigma \in \Sigma(\pi, \rho)} (|\sigma| - 1)!
				\right) f(\rho), \qquad \forall \pi \in \Pi_n
				\label{eq:mobius-2-g}
			\end{equation}
		\end{enumerate}
	\end{lemma}

	The set $\Sigma(\pi, \rho)$ may cause some confusion, so it is worth considering an 
	example before we proceed. Consider two partitions of $\Pi_5$: $\pi = \{ \{1\}, 
	\{2\}, \{4\}, \{3, 5\} \}$, and $\rho = \{ \{1, 2, 3, 5\}, \{4\} \}$. Clearly $\pi$ 
	is a refinement of $\rho$. Furthermore, we can obtain $\pi$ from $\rho$ by 
	partitioning each block of $\rho$. Let $\sigma_1$ represent the partition of $\{1, 2, 
	3, 5\}$ into $\{ \{1\}, \{2\}, \{3, 5\}\}$, and let $\sigma_2$ be the 
	partition of $\{4\}$ into $\{ \{4\} \}$. The set $\Sigma(\pi, \rho) = \{ \sigma_1, 
	\sigma_2 \}$ is the collection of these two partitions. 
	
	Next, we apply this new M\"obius inversion formula to invert \eqref{eq:bias}, 
	obtaining an unbiased estimator for products of sample moments:

	\begin{lemma}[Unbiased Estimation of Products of Sample Moments] \label{lem:unbiased}
		Consider the random vector $X = \begin{pmatrix} X_1 & X_2 & \cdots & X_n 
		\end{pmatrix}^\top$, and consider some multiset $[i_1, i_2, \dots, i_k]$ from the 
		indices $\{1, 2, \dots, n\}$. For every $\pi \in \Pi_k$, define a statistic
		\begin{equation}
			\hat m_\pi = \frac{1}{(N)_{|\pi|}} \sum_{\rho \ge \pi} \left(
				(-1)^{|\pi| - |\rho|}N^{|\rho|} \prod_{\sigma \in \Sigma(\rho, 
				\pi)} (|\sigma| - 1)! \right) \prod_{C \in \rho} \hat m_{[i_j \mid 
				j \in C]}
			\label{eq:mhat}
		\end{equation}
		where $B_1, B_2, \dots, B_{|\pi|}$ are the blocks of $\pi$. Then $\hat m_\pi$ is 
		an unbiased estimator of the product of raw moments $\prod_{B \in \pi} m_{[i_j 
		\mid j \in B]}$, i.e., 
		\begin{equation}
			\E[\hat m_\pi] = \prod_{B \in \pi} m_{[i_j \mid j \in B]}
			\label{eq:unbiased}
		\end{equation}
	\end{lemma}

	\begin{proof}
		Let us define two functions $f, g: \Pi_k \to \mathbb R$ by
		\begin{align*}
			f(\pi) &= N^{|\pi|} \E\left[
				\prod_{B \in \pi} \hat m_{[i_j \mid j \in B]}
			\right]  \\
			g(\pi) &= (N)_{|\pi|} \prod_{B \in \pi} m_{[i_j \mid j \in B]} 
		\end{align*}
		In terms of these functions, Lemma \ref{lem:bias} states that
		\[
			f(\pi) = \sum_{\rho \ge \pi} g(\rho), \qquad \forall \pi \in \Pi_k
		\]
		Therefore, applying the M\"obius inversion in Lemma \ref{lem:mobius-2}, we obtain 
		\eqref{eq:mobius-2-g}. Substituting in the definitions of $f$ and $g$ yields
		\begin{align*}
			(N)_{|\pi|} \prod_{B \in \pi} m_{[i_j \mid j \in B]}
			&= \sum_{\rho \ge \pi} (-1)^{|\pi| - |\rho|} \left(
				\prod_{\sigma \in \Sigma(\pi, \rho)} (|\sigma| - 1)! 
			\right) N^{|\rho|} \E \left[ \prod_{C \in \rho} \hat m_{[i_j \mid j \in C]} 
			\right] \\
			&= \E \left[ 
				\sum_{\rho \ge \pi} \left( (-1)^{|\pi| - |\rho|} N^{|\rho|}
				\prod_{\sigma \in \Sigma(\rho, \pi)} (|\sigma| - 1)! \right) \prod_{C \in 
				\rho} \hat m_{[i_j \mid j \in C]}
			\right] \\
			&= \E[ (N)_{|\pi|} \hat m_\pi]
		\end{align*}
		for all $\pi \in \Pi_k$, from which we immediately conclude \eqref{eq:unbiased}.
	\end{proof}

	We can now, at long last, use Lemma \ref{lem:unbiased} to prove that $k$-statistics 
	are unbiased estimators for cumulants.
	
	\begin{proof}[Proof (Theorem \ref{thm:k})]
		Substituting \eqref{eq:unbiased} into \eqref{eq:cum-as-mom}, we see that Theorem 
		\ref{lem:unbiased} and Lemma \ref{thm:cum-as-mom} together imply that
		\[
			\E\left[
				\sum_{\pi \in \Pi_k} (-1)^{|\pi| - 1} (|\pi| - 1)! \hat m_\pi
			\right] = \kappa_{[i_1, i_2, \dots, i_k]}
		\] 
		Expanding $\hat m_\pi$ using its definition \eqref{eq:mhat}, we obtain
		\[
			Q \triangleq \E \left[
				\sum_{\pi \in \Pi_k} 
					\frac{(-1)^{|\pi| - 1} (|\pi| - 1)!}{(N)_{|\pi|}}
				\sum_{\rho \ge \pi} \left(
				(-1)^{|\pi| - |\rho|} N^{|\rho|} \prod_{\sigma \in \Sigma(\rho, 
					\pi)} (|\sigma| - 1)! \right) \prod_{C \in \rho} \hat m_{[i_j \mid 
					j \in C]}
			\right]
			= \kappa_{[i_1, i_2, \dots, i_k]}
		\]
		where we have defined $Q$ as a placeholder for the expected value, for notational 
		compactness. The remainder of the proof is to simplify $Q$ down to $k_{[i_1, i_2, 
		\dots, i_k]}$ in \eqref{eq:k-stat}.
	
		The first thing to do is manipulate the sum, as follows:
		\begin{align*}
			Q &\triangleq \sum_{\pi \in \Pi_k} \sum_{\rho \ge \pi} 
				\frac{(-1)^{|\rho| - 1} (|\pi| - 1)! N^{|\rho|}}{(N)_{|\pi|}} 
				\left( \prod_{\sigma \in \Sigma(\rho, \pi)} (|\sigma| - 1)! \right)
				\prod_{C \in \rho} \hat m_{[i_j \mid j \in C]} \\
			&= \sum_{\rho \in \Pi_k} \sum_{\pi \le \rho} 
				\frac{(-1)^{|\rho| - 1} (|\pi| - 1)! N^{|\rho|}}{(N)_{|\pi|}} 
				\left( \prod_{\sigma \in \Sigma(\rho, \pi)} (|\sigma| - 1)! \right)
				\prod_{C \in \rho} \hat m_{[i_j \mid j \in C]} \\
			&= \sum_{\rho \in \Pi_k} (-1)^{|\rho| - 1} N^{|\rho|} \left(
				\sum_{\pi \le \rho} \frac{(|\pi| - 1)!}{(N)_{|\pi|}}
				\prod_{\sigma \in \Sigma(\rho, \pi)} (|\sigma| - 1)! 
			\right) \prod_{C \in \rho} \hat m_{[i_j \mid j \in C]}
		\end{align*}
		Because the refinements of $\rho$ are isomorphic to the product of lattices 
		$\Pi_{|C_1|} \times \Pi_{|C_2|} \times \cdots \times \Pi_{|C_{|\rho|}|}$ (where 
		each of the lattices corresponds to the partitions of a block of $\rho$), we can 
		replace the sum over $\pi \le \rho$ with a sum over tuples of partitions in this 
		product:
		\[
			R \triangleq \sum_{\pi \le \rho} \frac{(|\pi| - 1)!}{(N)_{|\pi|}}
			\prod_{\sigma \in \Sigma(\rho, \pi)} (|\sigma| - 1)! 
			= \sum_{\sigma_1 \in \Pi_{|C_1|}} \sum_{\sigma_2 \in \Pi_{|C_2|}} \cdots 
			\sum_{\sigma_{|\rho|} \in \Pi_{|C_{|\rho|}|}} 
			\frac{(\sum_{j = 1}^{|\rho|} |\sigma_j| - 1)!}{(N)_{\sum_{j=1}^{|\rho|} 
			|\sigma_j|}} \prod_{j = 1}^{|\rho|} (|\sigma_j| - 1)!
		\]
		Here $R$ is another placeholder for the middle term of this equation. Now, 
		observe that the dependence of this sum on $\sigma_j$ is entirely through the 
		size of $\sigma_j$. Furthermore, for a given size $b_j = |\sigma_j|$, there are 
		$\stirling{|C_j|}{b_j}$ partitions of size $b_j$ in $\Pi_{|C_j|}$. Therefore, we 
		simplify
		\[
			R = \sum_{b_1 = 1}^{|C_1|} \sum_{b_2 = 1}^{|C_2|} \cdots \sum_{b_{|\rho|} = 
			1}^{|C_{|\rho|}} \frac{(\sum_{j=1}^{|\rho|} b_j - 
			1)!}{(N)_{\sum_{j=1}^{|\rho|} b_j}} \left(
				\prod_{j=1}^{|\rho|} \stirling{|C_j|}{b_j} (b_j - 1)!
			\right)
			= \frac{c_\rho}{  N^{|\rho|} }
		\]
		invoking the definition of $c_\rho$ from \eqref{eq:coef}. Replacing $R$ with the 
		right side of this equation in our last expression for $Q$, we obtain
		\[
			Q = \sum_{\rho \in \Pi_k} (-1)^{|\rho| - 1} c_\rho \prod_{C \in \rho} \hat 
			m_{[i_j \mid j \in C]} = k_{[i_1, i_2, \dots, i_k]}
		\]
		which completes the proof.
	\end{proof} 

	\section{Computational Notes}
	\label{sect:computation}
	
	\begin{figure}
		\centering
		\includegraphics[width=0.45\linewidth]{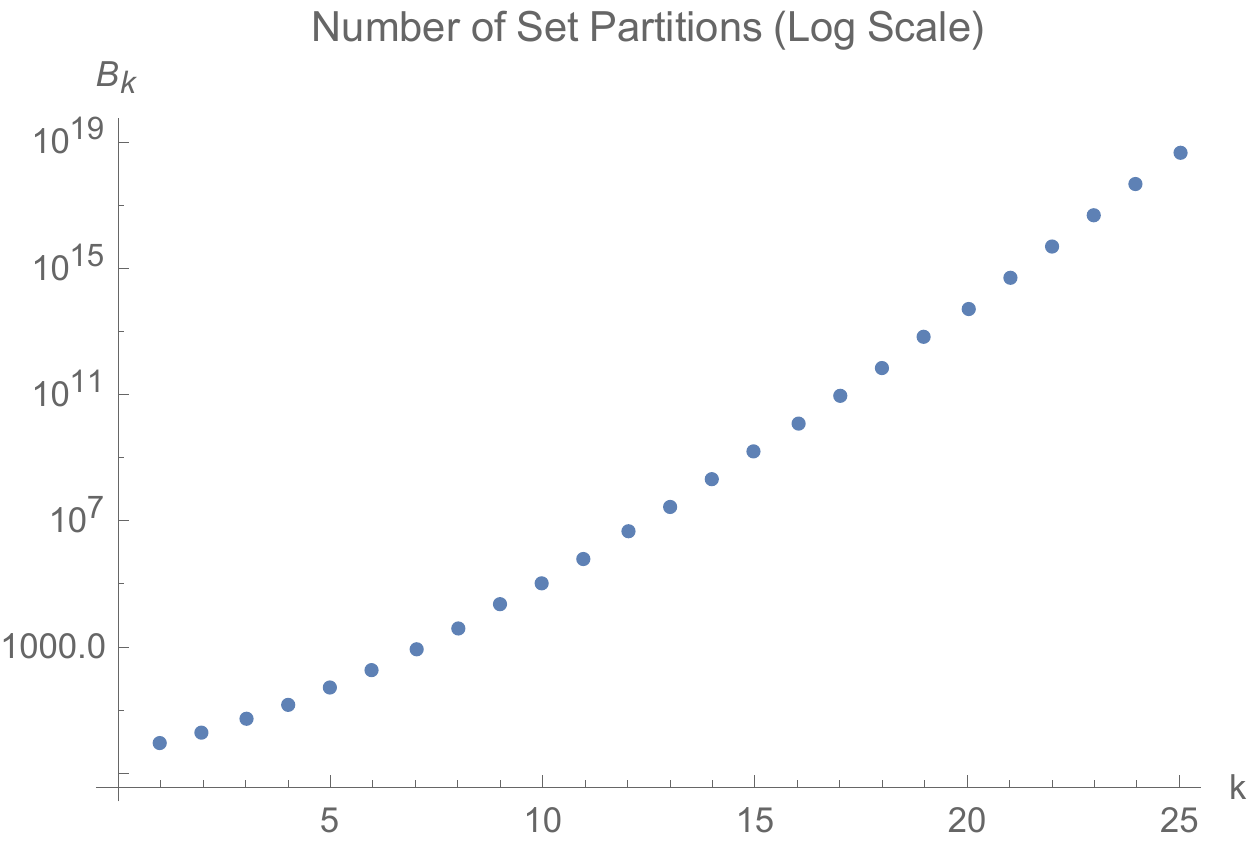} \hfill
		\includegraphics[width=0.45\linewidth]{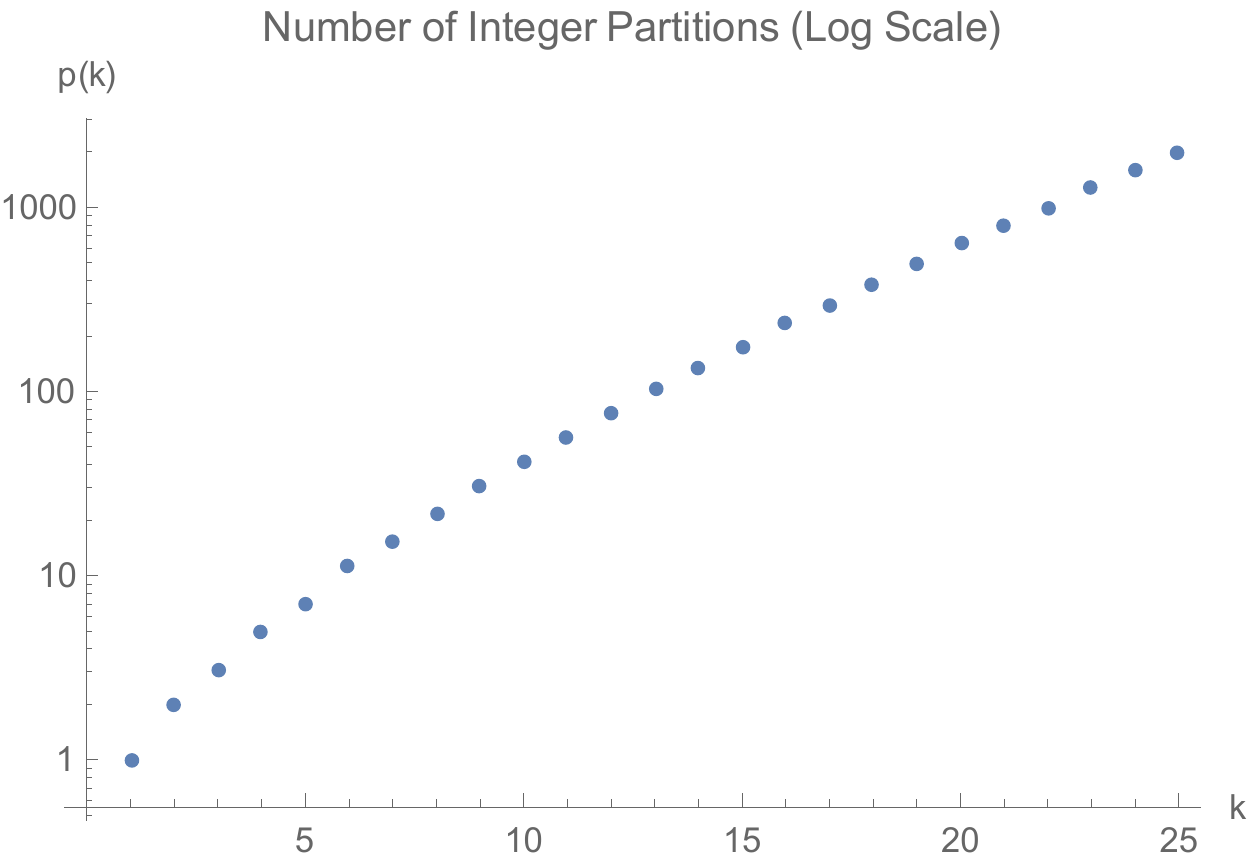}
		\caption{Plots related to the time complexity of computing \eqref{eq:k-stat}, 
		both on a log scale. The left plot shows Bell's number $B_k$, which reflects the 
		number of for loop iterations required to compute a $k$-statistic of order $k$. 
		The right plot shows the partition number $p(k)$, which counts the number of 
		unique values for the coefficient $c_\pi$.}
		\label{fig:nums}
	\end{figure}
	
	We will end our discussion of $k$-statistics with some comments on their computation. 
	The core of \eqref{eq:k-stat} is a for loop over the elements of $\Pi_k$, the 
	cardinality of which is given by Bell's number, $B_k$. The left plot in Figure 
	\ref{fig:nums} shows how Bell's number scales with $k$. At $k = 5$, the for loop only 
	needs to process 52 iterations, fast enough to perform repeated estimation within a 
	bootstrapping scheme, for example. By $k = 10$, the for loop must process almost 
	116,000 iterations---under a minute on a modern laptop, but certainly long enough to 
	make many repeated calculations cumbersome. By $k = 20$, we are up to roughly the 
	number of cells in the human body. Roughly speaking, fifth-order $k$-statistics are 
	about as high as one can reasonably go when performing many repeated calculations, 
	and tenth-order $k$-statistics are the threshold at which even one-off calculations 
	are too slow.\footnote{Another problem with high-order $k$-statistics is 
	statistical---the variance of the $k$-statistic scales poorly with order. As $k$ 
	increases, the necessary sample size $N$ increases very quickly, leading to slower 
	computation and infeasible data requirements. Variances of $k$-statistics are beyond 
	the scope of this document.}

	In order to make higher-order $k$-statistics more useful, a reasonable truncation of 
	the sum in \eqref{eq:k-stat} would be highly desirable. One possible approach is to 
	examine the formula in the large $N$ limit. Looking at \eqref{eq:coef} as $N \to 
	\infty$, we see that $c_\pi \to (|\pi| - 1)!$ since the $b_1 = b_2 = \cdots = 
	b_{|\pi|} = 1$ term of this sum dominates. Furthermore, the raw sample moment factors 
	converge on the sample moments (i.e., $\hat m_{[i_j \mid j \in C]} \to m_{[i_j \mid j 
	\in C]}$), which have no asymptotic dependence on $N$. With some prior knowledge of 
	the distribution, it may be possible to establish a hierarchy of summands $(|\rho| - 
	1)! \prod_{C \in \rho} m_{[i_j \mid j \in C]}$ in the large $N$ limit, allowing for a 
	corresponding truncation of the sum in \eqref{eq:k-stat}. Of course, this depends 
	highly on the raw moments of the distribution, and sufficient prior knowledge of 
	these moments may defeat the purpose of using $k$-statistics in the fist place.
	
	Instead of truncating the sum, another avenue to speed up the computation may be to 
	approximate the coefficients themselves, e.g., by assuming $c_\pi \approx (|\pi| - 
	1)!$ for large $N$. However, we suggest that approximating $c_\pi$ has little effect 
	on the efficiency, and that a better approach is to simply cache computed values of 
	the coefficients. The key is to observe that $c_\pi$ depends on the number and size 
	of each block, but not the blocks in and of themselves, so many terms of 
	the \eqref{eq:k-stat} sum will have identical value of the coefficients. As noted by
	\cite{EDN-GG:19}, the collections of block sizes in $\Pi_k$ actually correspond to 
	\textit{integer partitions} of $k$, so it is sufficient to compute and store one 
	value of $c_\pi$ per integer partition. The power of this trick lies in the fact that 
	the number of integer partitions (called the \textit{partition number}), $p(k)$, is 
	\textit{much} smaller than Bell's number $B_k$, at least when $k$ is of moderate size 
	or larger. The right plot in Figure \ref{fig:nums} shows how $p(k)$ scales with $k$ 
	much slower than $B_k$. For example, computing fifth-order $k$-statistics only 
	requires evaluating and storing $p(5) = 7$ unique values of $c_\pi$. Tenth-order 
	$k$-statistics involve $p(10) = 42$ unique values. And $k$-statistics of order 20, 
	which are intractable due to the size of $|\Pi_{20}|$, would require only 627 unique 
	values of the coefficients. In other words, the number of unique $c_\pi$ values is 
	very small compared to the size of $|\Pi_k|$, so precise evaluation and storage of 
	these coefficients is cheap. 
	
	We also note that \eqref{eq:k-stat} is easy to vectorize, i.e., it is straightforward 
	to evaluate $k$-statistics on several different samples simultaneously. This 
	vectorization is useful when evaluating $k$-statistics within a resampling scheme, 
	like jackknifing or bootstrapping. The computation is amenable to vectorization due 
	to the simple nature of the operations involved: linear combination (after computing 
	the $c_\pi$ coefficients), multiplication, and power sums are basic operations 
	that are supported in most libraries for array math.
	
	Finally, we take this moment to advertise \textit{PyMoments} \cite{KDS:20}, our own 
	Python library for computing multivariate $k$-statistics. \textit{PyMoments} 
	automatically caches the $c_\pi$ coefficients, saving them in a tree-based data 
	structure that can be re-used between different $k$-statistic evaluations (provided 
	that the sample size $N$ is the same). \textit{PyMoments} also supports vectorized 
	computation of $k$-statistics. Of course, we are open to feedback on how to improve 
	this library.

	\section{Conclusion}
	
	This document has provided an explicit expression for multivariate $k$-statistics, 
	allowing for unbiased estimation of multivariate cumulants. We were also able to 
	prove the lack of bias using fairly rudimentary combinatorics, and we provided a 
	light discussion on the computational aspect of $k$-statistics. It is our hope that 
	readers may be able to push some of the ideas of this paper forward into new, more 
	efficient algorithms and applications involving $k$-statistics.

	\bullobib
	

\begin{thebibliography}{10}

\bibitem{MA:97}
M.~Aigner.
\newblock {\em Combinatorial Theory}.
\newblock Springer, 1997.
\newblock \href {https://doi.org/10.1007/978-3-642-59101-3}
  {\path{doi:10.1007/978-3-642-59101-3}}.

\bibitem{PB:15}
P.~Bartlett.
\newblock {UCSB} {M}ath 116, {L}ecture notes: {M}\"obius inversion, 2015.
\newblock URL:
  \url{http://web.math.ucsb.edu/~padraic/ucsb_2014_15/math_116_s2015/math_116_s2015_lecture4.pdf}.

\bibitem{EDN-GG-DS:08}
E.~{Di~Nardo}, G.~Guarino, and D.~Senato.
\newblock A unifying framework for k-statistics, polykays and their
  multivariate generalizations.
\newblock {\em Bernoulli}, 14(2):440--468, 2008.
\newblock \href {https://doi.org/10.3150/07-BEJ6163}
  {\path{doi:10.3150/07-BEJ6163}}.

\bibitem{RAF-30}
R.~A. Fisher.
\newblock Moments and product moments of sampling distributions.
\newblock {\em Proceedings of the London Mathematical Society}, 2(1):199--238,
  1930.
\newblock \href {https://doi.org/10.1112/plms/s2-30.1.199}
  {\path{doi:10.1112/plms/s2-30.1.199}}.

\bibitem{MH:06}
M.~Hardy.
\newblock Combinatorics of partial derivatives.
\newblock {\em The Electronic Journal of Combinatorics}, 13, 2006.
\newblock URL: \url{https://arxiv.org/pdf/math/0601149.pdf}.

\bibitem{MGK-40}
M.~G. Kendall.
\newblock The derivation of multivariate sampling formulae from univariate
  formulae by symbolic operation.
\newblock {\em Annals of Eugenics}, 10(1):392--402, 1940.
\newblock \href {https://doi.org/10.1111/j.1469-1809.1940.tb02261.x}
  {\path{doi:10.1111/j.1469-1809.1940.tb02261.x}}.

\bibitem{PM-08}
P.~McCullagh.
\newblock {\em Tensor Methods in Statistics}.
\newblock Dover Publications, 2018.

\bibitem{EDN-GG:19}
E.~Di Nardo and G.~Guarino.
\newblock kstatistics: Unbiased estimators for cumulant products, 2019.
\newblock {R} package version 1.0.
\newblock URL: \url{https://CRAN.R-project.org/package=kStatistics}.

\bibitem{EDN-GG-DS:09}
E.~Di Nardo, G.~Guarino, and D.~Senato.
\newblock A new method for fast computing unbiased estimators of cumulants.
\newblock {\em Statistics and Computing}, 19(2):155, 2009.
\newblock \href {https://doi.org/10.1007/s11222-008-9080-0}
  {\path{doi:10.1007/s11222-008-9080-0}}.

\bibitem{CR-MDS:02}
C.~Rose and M.~D. Smith.
\newblock Mathstatica: {Mathematical} statistics with {Mathematica}.
\newblock In {\em Compstat}, pages 437--442. Springer, 2002.
\newblock \href {https://doi.org/10.1007/978-3-642-57489-4_66}
  {\path{doi:10.1007/978-3-642-57489-4_66}}.

\bibitem{GCR-JS:00}
G.-C. Rota and J.~Shen.
\newblock On the combinatorics of cumulants.
\newblock {\em Journal of Combinatorial Theory, Series A}, 91:283--304, 2000.
\newblock \href {https://doi.org/10.1006/jcta.1999.3017}
  {\path{doi:10.1006/jcta.1999.3017}}.

\bibitem{KDS:20}
K.~D. Smith.
\newblock {PyMoments}: {A} {Python} toolkit for unbiased estimation of
  multivariate statistical moments, 2020.
\newblock URL: \url{https://github.com/KevinDalySmith/PyMoments}.

\bibitem{TPS-83}
T.~P. Speed.
\newblock Cumulants and partition lattices.
\newblock {\em Australian Journal of Statistics}, 25(2):378--388, 1983.
\newblock \href {https://doi.org/10.1111/j.1467-842X.1983.tb00391.x}
  {\path{doi:10.1111/j.1467-842X.1983.tb00391.x}}.

\end{thebibliography}
\end{document}